\documentclass[11pt]{article}
\usepackage{amsmath, amsthm, amssymb,color,fullpage,url,booktabs}  
\usepackage[colorlinks,pdftex]{hyperref} 

\newcommand{\nc}{\newcommand}
\nc{\rnc}{\renewcommand}

\newcommand{\bra}[1]{\left\langle #1\right|}
\newcommand{\ket}[1]{\left|#1\right\rangle}
\newcommand{\proj}[1]{\left|#1\right\rangle\left\langle #1\right|}
\newcommand{\braket}[2]{\left\langle #1\middle|#2\right\rangle}

\DeclareMathOperator{\Img}{Im}

\DeclareMathOperator{\tr}{tr}

\DeclareMathOperator{\Span}{span}
\DeclareMathOperator{\supp}{supp}
\DeclareMathOperator{\swap}{SWAP}

\def\be#1\ee{\begin{equation}#1\end{equation}}
\def\bea#1\eea{\begin{eqnarray}#1\end{eqnarray}}
\def\beas#1\eeas{\begin{eqnarray*}#1\end{eqnarray*}}
\def\ba#1\ea{\begin{align}#1\end{align}}
\def\bas#1\eas{\begin{align*}#1\end{align*}}
\def\bpm#1\epm{\begin{pmatrix}#1\end{pmatrix}}
\def\non{\nonumber}

\def\eq#1{(\ref{eq:#1})}

\def\L{\left} 
\def\R{\right}
\def\ra{\rightarrow}
\def\ot{\otimes}

\newtheorem{thm}{Theorem}
\newtheorem*{thm*}{Theorem}

\newtheorem{cor}[thm]{Corollary}
\newtheorem{lem}[thm]{Lemma}
\newtheorem{prop}[thm]{Proposition}
\newtheorem{dfn}{Definition}
\newtheorem{proto}{Protocol}

\makeatletter
\newtheorem*{rep@theorem}{\rep@title}
\newcommand{\newreptheorem}[2]{%
\newenvironment{rep#1}[1]{%
 \def\rep@title{#2 \ref{##1} (restatement)}%
 \begin{rep@theorem}}%
 {\end{rep@theorem}}}
\makeatother

\newreptheorem{thm}{Theorem}
\newreptheorem{lem}{Lemma}

\def\eps{\epsilon}

\def\cI{{\cal I}}

\def\cM{{\cal M}}
\def\cN{\mathcal{N}}

\def\cS{\mathcal{S}}

\def\cU{\mathcal{U}}

\def\bbC{\mathbb{C}}

\DeclareMathOperator*{\bbE}{\mathbb{E}}

\def\bbN{\mathbb{N}}
\def\bbR{\mathbb{R}}
\def\bbZ{\mathbb{Z}}

\def\benum{\begin{enumerate}}
\def\eenum{\end{enumerate}}
\def\bit{\begin{itemize}}
\def\eit{\end{itemize}}

\newcommand{\secref}[1]{Section~\ref{sec:#1}}

\newcommand{\lemref}[1]{Lemma~\ref{lem:#1}}
\newcommand{\thmref}[1]{Theorem~\ref{thm:#1}}
\newcommand{\propref}[1]{Proposition~\ref{prop:#1}}

\newcommand{\corref}[1]{Corollary~\ref{cor:#1}}

\usepackage{cancel}

\nc{\ssym}[2]{\vee^{#2}\bbC^{#1}}
\nc{\psym}[2]{P_{\text{sym}}^{#1,#2}}
\nc{\dimsym}[2]{#1[#2]}
\nc{\dsym}[1]{\dimsym{d}{#1}}
\DeclareMathOperator{\clone}{Clone}
\DeclareMathOperator{\MP}{MP}

\begin{document}

\title{The Church of the Symmetric Subspace}
\author{Aram W. Harrow\thanks{Center for Theoretical Physics,
    MIT.  email: {\tt aram@mit.edu}}}
\date{\today}
\maketitle
\begin{abstract}
The symmetric subpace has many applications in quantum information
theory.  This review article begins by explaining key  background
facts about the symmetric subspace from a quantum information
perspective.  Then we review, and in some places extend, work of
Werner and Chiribella that connects the symmetric subspace to 
state estimation, optimal cloning, the de Finetti theorem and other
topics.  In the third and final section, we discuss how the symmetric
subspace can yield concentration-of-measure results via the
calculation of higher moments of random quantum states.

There are no new results in this article, but only some new proofs of
existing results, such as a variant of the exponential de Finetti theorem.  The
purpose of the article is (a) pedagogical, and (b) to collect in one
place many, if not all, of the quantum information applications of the
symmetric subspace. 
\end{abstract}

\tableofcontents
\vspace{5mm}

Schur-Weyl duality between the unitary and symmetric groups is a
powerful and useful tool in quantum 
information.  But some aspects of it are unsatisfactory.  The proofs
are rarely fully self-contained, and require excursions into other Lie
algebras.  At the same time, they involve irreducible representations
(irreps) that lack simple, explicit, constructions, making the theory
less useful for calculations than one would like.  But in many cases,
the symmetric subspace is the only necessary piece that needs to be
understood.  The symmetric subspace is the simplest component
of Schur-Weyl duality (with the antisymmetric subspace a close second)
and often can be used effectively without the need to ever explicitly
invoke representation theory.

In \secref{sym} of these notes, I will give a self-contained review of the properties
of the symmetric subspace.  Some applications of the symmetric
subspace involve cloning, state estimation and the de Finetti
theorem.  These are discussed in a unified way by \cite{Chiri10}, and
in \secref{chiri}, I will give a brief review of that paper.  Another
reason to study the symmetric subspace is that it is a way of looking
at higher moments of quantum states.  In \secref{conc}, I'll explain
how this can be used to give alternate and unified derivations of many
concentration-of-measure results in quantum information theory.  This
is the only part of the paper to mostly consist of original work, although
even here this consists mostly of new proofs of previously known theorems.

\begin{table}
\begin{center}
\begin{tabular}{|l|l|}
\hline
Variable & Definition \\ \hline
$d$ &  local dimension of each subsystem\\
$[d]$ & the set $\{1,\ldots,d\}$\\
$\ssym{d}{n}$ & the symmetric subspace of $(\bbC^d)^{\ot n}$ \\
$\psym{d}{n}$ & the orthogonal projector onto $\ssym{d}{n}$\\
$\dsym{n}$ & $\binom{d+n-1}{n} = \dim\ssym{d}{n} = \tr\psym{d}{n}$\\
$P_d(\pi)$ & $\sum_{i_1,\ldots,i_n\in [d]} \ket{i_{\pi^{-1}(1)}, \ldots, i_{\pi^{-1}(n)}}\bra{i_1,\ldots,i_n}$\\
$\cS_n$ & the symmetric group on $n$ objects \\
$\cU_d$ & the group of $d\times d$ unitary matrices \\
$\bbZ_+$ & nonnegative integers \\
$\cI_{d,n}$ & $\{(t_1,\ldots,t_d) : t_1,\ldots,t_d\in \bbZ_+, t_1+\ldots+t_d = n\}$\\
$\binom{n}{\vec t}$ & $\frac{n!}{t_1!\ldots t_d!}$ \\
L(V) & linear operators on a vector space $V$ \\
H(V) & Hermitian operators on $V$ \\
$\varphi$ & $\proj \varphi$ (convention used for all pure states)\\
\hline\end{tabular}
\caption{Here is a table of notation, used throughout the notes.  For
  now, you should skip it and go straight to \secref{sym}.}
\end{center}
\end{table}

\section{The symmetric subspace}\label{sec:sym}

One motivation for writing these notes is that there is no
comprehensive treatment of the symmetric subspace from the quantum
information viewpoint.  Ref.~\cite{matthias} covers some of it,
Ref.~\cite{Har05} a little less, and Refs.~\cite{GW98,Stanley-EC2} are
excellent, but approach the subject respectively from the
Lie-algebraic or combinatorial perspective, rather than in terms of
quantum information.   All of these are
really more focused on Schur-Weyl duality in general than the
symmetric subspace specifically.  Some exceptions are
Ref.~\cite{BBDEJM96,JHHH98}, which are good, but present only as much of the
theory as they need for for their applications.   Koenraad Audenaert
has also written some nice notes on the representation theory of the
symmetric group~\cite{Audenaert06}.

Let $\cS_n$ be the symmetric group on $n$ letters.  For $\pi\in\cS_n$, define 
$$P_d(\pi) = \sum_{i_1,\ldots,i_n\in [d]} \ket{i_{\pi^{-1}(1)}, \ldots, i_{\pi^{-1}(n)}}\bra{i_1,\ldots,i_n}.$$
Note that $P_d(\pi_1\pi_2) = P_d(\pi_1)P_d(\pi_2)$.  In other words, $P_d$ is a representation of $\cS_n$ on $(\bbC^d)^{\ot n}$.  

The {\em symmetric subspace} of $(\bbC^d)^{\ot n}$ is denoted $\ssym d n$ and is defined to be 
\be \ssym d n = \{\ket{\psi}\in(\bbC^d)^{\ot n} : P_d(\pi)\ket\psi = \ket \psi \,\forall\ket\psi\in\cS_n\}. \ee
(The $\vee$ denotes the symmetric product, by contrast with $\wedge$ which stands for the antisymmetric product, and which we will not discuss here.)

Define 
\be \psym d n = \frac{1}{n!} \sum_{\pi\in\cS_n}P_d(\pi).\label{eq:psym-group-avg}\ee
\begin{prop}\label{prop:psym-proj}
$\psym d n$ is the orthogonal projector onto $\ssym d n$.  
\end{prop}

\begin{proof}
Since group multiplication is invertible, we have that for any $\pi\in\cS_n$
\be\begin{split} P_d(\pi)\psym d n &= P_d(\pi)\frac{1}{n!} \sum_{\pi'\in\cS_n} P_d(\pi')
\\&= \frac{1}{n!} \sum_{\pi'\in\cS_n} P_d(\pi\pi')
\\&= \frac{1}{n!} \sum_{(\pi^{-1}\pi')\in\cS_n} P_d(\pi')
\\&= \frac{1}{n!} \sum_{\pi'\in\cS_n} P_d(\pi')
= \psym d n.
\end{split}\label{eq:psym-eats-pi}\ee
Similarly, we have $\psym d n P_d(\pi) = \psym d n$.

This implies that 
$$(\psym d n)^\dag \psym d n = \frac{1}{n!} \sum_{\pi\in\cS_n} P_d(\pi^{-1}) \psym d n = 
\frac{1}{n!} \sum_{\pi\in\cS_n}\psym d n = \psym d n.$$
Therefore $\psym d n $ is an orthogonal projector, since $\Pi^\dag
\Pi=\Pi$ is a necessary and sufficient condition for an operator $\Pi$
to be an orthogonal projector.

We also use \eq{psym-eats-pi} to show that for any $\ket\psi\in(\bbC^d)^{\ot n}$, 
$$P_d(\pi)\psym d n \ket\psi = \psym d n\ket\psi.$$
Thus $\psym d n\ket\psi \in \ssym d n $, and we have that $ \Img\psym d n\subseteq  \ssym d n$.

To show  that $\ssym d n \subseteq \Img\psym d n$, we observe that if $\ket\psi \in \ssym d n$ then $\psym d n \ket \psi = \frac{1}{n!} \sum_{\pi\in\cS_n}P_d(\pi)\ket \psi = \ket \psi$.
\end{proof}

The alert reader will notice that almost no properties of $\cS_n$ were
used in the above proof.  We can generalize 
\propref{psym-proj} to a large class of groups.  The necessary condition is that a
group $G$ should have an invariant measure $\mu$.  That is, for any
integrable function $f:G\ra \bbC$, and any $g\in G$, we have
$\int_{x\in G} \mu(x)f(x) {\rm d}x = \int_{x\in G} \mu(x)f(gx) {\rm
  d}x$.  Such measures exist for all finite groups (take $\mu(x) = 1/|G|$ and replace the integral by a sum)
and for all compact Lie groups, such as the unitary group. In the
latter case, there is a unique measure (up to normalization) called
the Haar measure.

  For a vector space $V$, define $L(V)$ to be the set of linear operators on $V$. 
\begin{prop} \label{prop:invar-proj}
Let $G$ be a group with an invariant measure $\mu$, and a
  representation $R : G \ra L(V)$.  Define 
\ba V^G &:= \{\ket\psi\in V :   R(g)\ket\psi = \ket \psi \forall g\in
G\} \\
\Pi & := \int_{x\in G} {\rm d}x \mu(x) R(x)\ea
Then $\Pi$ is an orthogonal projector onto $V^G$.
\end{prop}
We omit the proof, as it follows the same lines as that of \propref{psym-proj}.

We now return to our discussion of the symmetric subspace, and give
two equivalent characterizations of $\ssym d n$. First define 
$$A =\Span \{\ket{\varphi}^{\ot n} : \ket{\varphi}\in\bbC^d\}.$$

Second, let $\bbZ_+$ denote nonnegative integers. Let $\cI_{d,n} = \{(t_1,\ldots,t_d) : t_1,\ldots,t_d\in \bbZ_+, t_1+\ldots+t_d = n\}$.  For $\vec t\in \cI_{d,n}$ we abbreviate the multinomial coefficient $\frac{n!}{t_1!\ldots t_d!}$ by $\binom{n}{\vec t}$.  For $\vec i = (i_1,\ldots,i_n)\in [d]^n$, the {\em type} of $\vec i$ is denoted $T(\vec i)$ and defined to be the vector in $\cI_{d,n}$ whose $j^{\text{th}}$ entry is the number of times that $j$ appears in the string $(i_1,\ldots,i_n)$.  Note that $|T^{-1}(\vec t)| = \binom{n}{\vec t}$.
Now define
$$\ket{s_{\vec t}} := \sqrt{\binom{n}{\vec t}} \sum_{\vec i: T(\vec i) =\vec t} \ket{i_1,\ldots,i_n}.$$
Finally we can define the subspace
$$ B = \Span \{\ket{s_{\vec{t}}} : \vec t \in \cI_{d,n}\}.$$

We can now state our main theorem about the structure of $\ssym d n$.
\begin{thm}\label{thm:sym}
$$\ssym d n = A = B.$$
\end{thm}

\begin{proof}
Since $A$ and $B$ are both spanned by sets of vectors that are individually invariant under $P_d(\pi)$, it follows that $A\subseteq \ssym d n$ and $B\subseteq \ssym d n$.

To show that $\ssym d n\subseteq B$, we note that $\psym d n \ket{i_1,\ldots, i_n} = \binom{n}{T(\vec i)}^{-1/2} \ket{s_{T(\vec i)}}$ and therefore $\Img \psym d n \subseteq B$. Since $\ssym d n = \Img \psym d n$, we conclude that $\ssym d n = B$.

The last step is to show that $B\subseteq A$.  Here we use polynomials
in a clever way.  Suppose  $p(x) = v_0 + x v_1 + \ldots + x^d v_d$,
for $v_0,\ldots, v_d$ vectors in a finite-dimensional space $V$.  For
some subspace $W\subset V$, suppose that $p(x)\in W$ for all $x$.
Then I claim that $v_0,\ldots,v_d\in W$.  The proof is that derivatives of $p(x)$ can be expressed as limits of linear combinations of $p(x)$ for different values of $x$, and therefore are all contained in $W$.   Then we use the fact that $v_k =\frac{1}{k!} \frac{\partial^k}{\partial x^k}p(x)|_{x=0}$.

By induction on the number of variables, we can extend this to argue
that if $p(x_1,\ldots,x_d)\in W$ for all $x_1,\ldots,x_d$ then the
coefficient of $x_1^{t_1}\ldots x_d^{t_d}$ must be in $W$ for each
$t_1,\ldots,t_d\in \bbZ_+^d$.

Now, consider the polynomial $\ket{p(x_1,\ldots,x_d)}:=(\sum_{i=1}^d
x_i \ket{i})^{\ot n}$.  Since $\ket{p(x_1,\ldots,x_d)}$ is a tensor
power state, it belongs to $A$.  The coefficient of $x_1^{t_1}\ldots
x_d^{t_d}$ in $\ket{p(x_1,\ldots,x_d)}$ is proportional to
$\ket{s_{\vec t}}$.  Therefore $\ket{s_{\vec t}}\in A$ for all $\vec
t\in\cI_{d,n}$. We conclude that $B\subseteq A$.
\end{proof}

One interesting consequence of \thmref{sym} is obtained by replacing
$\bbC^d$ with $M_d$, the vector space of $d\times d$ matrices.  The
symmetric group acts on $M_d^{\ot n}$ by conjugation, with $\pi$ sending $M\in
M_d^{\ot n}$ to $P_d(\pi) M P_d(\pi)^\dag$. 
\begin{cor}\label{cor:commutant}
If $M\in M_d^{\ot n}$, then $[M, P_d(\pi)]=0$ for all $\pi\in \cS_n$
if and only if $M\in\Span\{X^{\ot n}: X\in M_d\}$.
\end{cor}

\corref{commutant} is like a baby de Finetti theorem (of which we will
discuss more in \secref{chiri}).  It says that a
permutation-invariant state $\rho$ can be written as $\sum a_i
X_i^{\ot n}$.  Unfortunately, $a_i$ and $X_i$ do not have to be
positive, making this decomposition less useful.  

Another natural way to understand $\ssym d n$ is in terms of representation theory.
\begin{thm}\label{thm:sym-irrep}
$\ssym d n$ is an irreducible represention of $\cU_d$ under the action
$U\mapsto U^{\ot n}$.
\end{thm}
This fact is often proved using facts about the irreducible
representations of Lie algebras.  To keep things self-contained, I
will give an elementary proof using ideas familiar to quantum
information.
\begin{proof}
Consider an arbitrary pair of unit vectors $\ket{\psi_1},\ket{\psi_2}\in
\ssym d n$.  We will demonstrate the existence of a $U\in \cU_d$ such
that $\bra{\psi_1} U^{\ot n}\ket{\psi_2} \neq 0$.   Equivalently, we will
choose a probability distribution over $U$ satisfying 
\be \bbE_U \bra{\psi_1} U^{\ot n}\ket{\psi_2} \neq 0
\label{eq:not-irred}.\ee

To this end, choose unit vectors $\ket{\varphi_1},\ket{\varphi_2}\in
\bbC^d$ such that $\bra{\psi_i}\cdot \ket{\varphi_i}^{\ot n}\neq 0$
for $i=1,2$. By \thmref{sym}, these vectors must exist.   Then, for $i=1,2$,
choose $V_i$ uniformly at random from the set of unitaries satisfying
$V_i\ket{\varphi_i}=\ket{\varphi_i}$.  Such unitaries can be
constructed by choosing a random element of $\cU_{d-1}$ and embedding
it in the space orthogonal to $\ket{\varphi_i}$.
Since $\bbE_{V_i} V_i^{\ot n}$ is an average over a group action, \propref{invar-proj} implies that it is a
projector onto the set of vectors fixed by each $V_i^{\ot n}$; in particular, 
$$\bbE_{V_i} V_i^{\ot n} =\proj{\varphi_i}^{\ot n}.$$
Finally, choose $W\in \cU_d$ to be a unitary satisfying
$W\ket{\varphi_2}=\ket{\varphi_1}$, and let $U = V_1^\dag W V_2$.
Then $\bbE_U U^{\ot n}= \ket{\varphi_1}\bra{\varphi_2}^{\ot n}$, and
since we have assumed $\ket{\psi_i}$ has nonzero overlap with
$\ket{\varphi_i}^{\ot n}$ for $i=1,2$, we obtain \eq{not-irred}.
\end{proof}

Using \thmref{sym-irrep} and Schur's Lemma gives us another
characterization of $\psym d n$.  
\begin{prop}\label{prop:twirl}
\be\bbE_\varphi \varphi^{\ot n} = \frac{\psym d n}{\tr \psym d n} =
  \frac{\psym d n}{\dsym n} =
  \frac{\sum_{\pi\in\cS_n}P_d(\pi)}{d(d+1)\cdots (d+n-1)}.
\label{eq:psym-twirl}\ee
\end{prop}
Here $\bbE_{\varphi}$ means that we average over a randomly chosen
unit vector $\ket\varphi \in \bbC^d$.  
\begin{proof}
Observe that $\rho:=\bbE_\varphi \varphi^{\ot n}$ commutes with all $U^{\ot
  n}$.   Thus, by Schur's Lemma and \thmref{sym-irrep}, $\rho$ must be
proportional to the identity operator on that space, which is $\psym d
n$.  To find the normalization, we observe that $\tr\rho=1$.
\end{proof}

One consequence of \propref{twirl} is that 
\be \{\dsym n \proj \varphi^{\ot n}\,
d\varphi\} \label{eq:cov-POVM}\ee
 forms a continuous POVM.

Another consequence is that averaging $\hat\varphi^{\ot n}$ over Gaussian vectors
$\ket{\hat\varphi}$ gives an operator proportional to a projector.  If the
normalization (and covariance) of the Gaussian is chosen so that $\bbE_{\hat\varphi}
\hat\varphi =  I/d$, then we can show that
\be \bbE_{\hat\varphi}\hat\varphi^{\ot n} 
=  \frac{n!}{d^n} \psym d n = d^{-n}
\sum_{\pi\in\cS_n} P_d(\pi).\label{eq:Gaussian}\ee
How?  Well, let $\ket{\hat\varphi} = x \ket\varphi$, for $x\in \bbR_+$
and $\ket\varphi$ a unit vector.  Because of the rotational invariance
of the Gaussian distribution, it follows that $x$ and $\ket\varphi$
are independent random variables.  Thus
$\bbE_{\hat\varphi}\hat\varphi^{\ot n} = \bbE_x |x|^{2n} \bbE_\varphi
\varphi^{\ot n}$.  It remains only to compute $\bbE_x |x|^{2n}$.    To
do so, let $a_j := \braket{j}{\hat\varphi}$ so that $|x|^2 = |a_1|^2 +
\cdots + |a_d|^2$.  Next, we recall the formula for a Gaussian integral:
\be
\int_{a \in \bbC^d} {\rm d}a\, e^{-\alpha(|a_1|^2 + \cdots + |a_d|^2)}
 = (\pi/\alpha)^d.\label{eq:gauss-integral}\ee
We can use this to calculate 
$\bbE_x |x|^{2n}$ by differentiating \eq{gauss-integral} with respect
to $-\alpha$ $n$ times, then dividing by the normalization
$(\pi/\alpha)^d$, and finally setting $\alpha=d$.
This yields $\bbE_x |x|^{2n} = (1+\frac{1}{d})\cdots(1+\frac{n-1}{d})$.
Combining with \eq{psym-twirl}, we obtain \eq{Gaussian}.

Alternatively, \eq{Gaussian} can be derived directly (and is sometimes
called Wick's theorem\footnote{The way to do this calculation is to
  calculate the integral of $\exp(-\sum_{i=1}^d \alpha_i |a_i|^2)$ and
  differentiate with respect to various $\alpha_i$.}), and then used
to obtain \eq{psym-twirl}.  See \cite{folland01} for a nice exposition
of this approach.

\subsection{Operators on the symmetric subspace}
For a complex vector space $V$, let $L(V)$ denote the space of operators on
$V$ and $H(V)$ the space of Hermitian operators on $V$.    Given that
$\ssym d n $ is spanned by vectors of the form $\ket\varphi^{\ot n}$,
can we say something similar about $L(\ssym d n)$ and $H(\ssym d n)$?  

Happily, in this case the matrices $\varphi^{\ot n}$ play the same
role, and we have
\begin{subequations}\label{eq:op-sym}\ba
L(\ssym d n) &= \Span_\bbC \{\varphi^{\ot n} : \ket\varphi\in\bbC^d\}
\label{eq:lin-sym}
\\
 H(\ssym d n) &= \Span_\bbR \{\varphi^{\ot n} : \ket\varphi\in\bbC^d\}
\label{eq:herm-sym}.\ea\end{subequations}

In both cases, the RHS is trivially contained in the LHS.  
Conversely, the LHS of \eq{lin-sym} can be expressed as a span of
operators of the form $(\ket\alpha\bra\beta)^{\ot n}$.  To express
$(\ket\alpha\bra\beta)^{\ot n}$ as a linear combination of terms on the RHS of \eq{lin-sym},
define $\ket{v_{x,y}} = e^{ix}\ket\alpha + e^{iy}\ket\beta$ and note
\be
(\ket\alpha\bra\beta)^{\ot n} = 
\frac{1}{(2\pi)^2}\int_0^{2\pi} {\rm d}x \int_0^{2\pi} {\rm d}y\, 
e^{in(y-x)}
(\ket{v_{x,y}}\bra{v_{x,y}})^{\ot n}
\ee
Similarly for $H(\ssym d n)$, the LHS is spanned by the operators
$(\ket\alpha\bra\beta)^{\ot n} + (\ket\beta\bra\alpha)^{\ot n}$.  We
write this operator as a real linear combination of terms from the RHS of
\eq{herm-sym} as follows:
\be
(\ket\alpha\bra\beta)^{\ot n}  + (\ket\beta\bra\alpha)^{\ot n}= 
\frac{1}{(2\pi)^2}\int_0^{2\pi} {\rm d}x \int_0^{2\pi} {\rm d}y\, 
(e^{in(y-x)}+e^{-in(y-x)})
(\ket{v_{x,y}}\bra{v_{x,y}})^{\ot n}
\ee

\subsection{The real case}

What about $\vee^n \bbR^d$?   Things are now totally different.   Let
$\ket\gamma\in \bbR^d$ be a random unit vector and
$\ket{\hat\gamma}\in\bbR^d$ be a random Gaussian vector with $\bbE
\braket{\hat\gamma}{\hat\gamma}=1$.  
To describe $\bbE \hat\gamma^{\ot n}$, we introduce some more
notation.  Let $\cM_{2n}$ be the set of perfect matchings on $[2n]$;
i.e. $n$ disjoint subsets of $[2n]$, each containing two elements.
Say that a string $i_1,\ldots,i_{2n}\in [d]$ is compatible with
$M\in\cM_{2n}$ if $i_j=i_k$ for each $\{j,k\}\in M$.  Let $S_M$ denote
the set of $i_1,\ldots,i_{2n}$ that are compatible with $M$, and
define $\sigma_M := \sum_{(i_1,\ldots,i_{2n})\in S_M}
\ket{i_1,\ldots,i_n}\bra{i_{n+1},\ldots,i_{2n}}$.
Note that $P_d(\pi) = \sigma_{M_\pi}$ where we define $M_\pi :=
\{(1,n+\pi(1)),\ldots,(n,n+\pi(n))\}$; however, other matchings do not
correspond to any permutation.
The moments of
$\hat\gamma$ are then given by
\be\bbE \hat\gamma^{\ot n} = d^{-n} \sum_{M\in\cM_{2n}} \sigma_M.
\ee
I'll skip the derivation, as it's similar to the complex case.

As an example, when $n=2$, then
$\bbE \hat\gamma^{\ot 2} = \frac{I + \swap}{d^2} + \frac{\Phi}{d}$, where
$\ket{\Phi} = d^{-1/2}\sum_{i=1}^d \ket{i,i}$ is a maximally entangled
state.  The $\ket\Phi$ term is new, and dramatically increases the
largest eigenvalue of resulting matrix.  To see why it appears,
consider a simple (univariate) Gaussian variable $x$.  If $x$ is a
complex Gaussian, then $\bbE x^2=\bbE \bar{x}^2=0$, and only $\bbE x\bar
x$ is nonzero.  However, if $x$ is a real Gaussian, then $\bbE x^2$ is
nonzero.  This means that there additional terms, corresponding to matchings not of the form $M_\pi$, that contribute to terms like
$\Phi$.  The reason that these terms lead to higher eigenvalues is
that they don't distinguish between row and column indices, and so are
not adapted to the matrix structure.

For unit vectors, the situation is similar except for the overall
normalization, which is described by a higher moment of a
$\chi^2$-distribution.  After a (skipped) calculation, one obtains
\be \bbE  \gamma^{\ot n} =
\L(\frac{d}{2}\R)^n \frac{\Gamma(d/2)}{\Gamma(n+d/2)}
\bbE \hat\gamma^{\ot n} =
\frac{1}{2^n}\frac{\Gamma(d/2)}{\Gamma(n+d/2)}
 \sum_{M\in \cM_{2n}} \sigma_M.\ee
Here $\Gamma(z)$ is the gamma function, equal to $z-1!$ for integer
$z$, and $\frac{(2z-1)!}{4^z(z-1/2)!}\sqrt{4\pi}$ for half-integer $z$.
As in the complex case, random unit vectors resemble random Gaussians
when $n$ is small relative to $d$.

\section{Estimation, cloning and the de Finetti theorem}\label{sec:chiri}
This section discusses three important applications of the symmetric
subspace, essentially following the treatment of \cite{Chiri10}, but
with some of the material on cloning from \cite{Werner98}.

Consider the following three problems.
\benum
\item {\em State estimation:} Measure $\ket{\varphi}^{\ot n}$ to obtain
  an estimate $\ket{\hat\varphi}$. Try to maximize
  $\bbE_{\varphi}|\braket{\varphi}{\hat\varphi}|^{2k}$ for some $k$.
\item {\em Cloning:} Construct a map $T$ from $n$ qudits to $n+k$ qudits
  that maximizes $\bbE_\varphi \tr \varphi^{\ot n+k} T(\varphi^{\ot
    n})$.
\item {\em de Finetti:} Given $\ket{\psi}\in\ssym d n$, how well is
  $\tr_{n-k}\psi$ approximiated by a mixture of tensor power states?
\eenum

It turns out that there are close relations between these problems.
\subsection{Estimation and Measure-and-prepare channels}
Start with state estimation.   Note that
$|\braket{\varphi}{\hat\varphi}|^{2k} = \tr \varphi^{\ot k}
\hat\varphi^{\ot k}$.  The most general strategy possible (more or
less) is to perform the 
POVM with measurement operators $M_1,\ldots,M_\ell$ (with $M_1 +
\cdots + M_\ell = \psym d n$), and upon outcome
$i$, to output the estimate $\ket{\hat\varphi_i}$.  Let $\rho_i =
\hat\varphi_i^{\ot k}$.  Then
\ba F_{\text{estimate}} &= 
\bbE_\varphi \sum_{i=1}^\ell \tr (\varphi^{\ot n}M_i) \tr (\varphi^{\ot
  k}\rho_i)
\\ &= 
\bbE_\varphi \sum_{i=1}^\ell \tr \varphi^{\ot n+k}(M_i \ot \rho_i)
\\ &= 
\tr \frac{\psym d {n+k}}{\dsym {n+k}} \sum_{i=1}^\ell (M_i \ot \rho_i)
\label{eq:replace}
\\ &\leq  
\frac{\sum_{i=1}^\ell \tr(M_i \ot \rho_i)}{\dsym{n+k}} 
\\ &=
\frac{\dsym n}{\dsym{n+k}} 
\label{eq:est-bound}
\ea
On the other hand, \eq{est-bound} is achieved by using the continuous
POVM from \eq{cov-POVM}.  Why?   We replace \eq{replace} with
\be
\tr \frac{\psym d {n+k}}{\dsym {n+k}} \bbE_{\hat\varphi} \dsym n
\hat\varphi^{\ot n} \ot \hat\varphi^{\ot k} = 
\frac{\dsym n}{\dsym{n+k}}.\ee

This analysis has also yielded the solution to a related problem,
which is to find the optimal ``measure-and-prepare'' channel mapping
$\ket{\varphi}^{\ot n}$ to an approximation of $\ket{\varphi}^{\ot
  k}$.  Measure-and-prepare channels are of the form 
\be T(\sigma) =
\sum_i \tr (M_i \sigma) \rho_i\label{eq:MP}\ee
 and are  also called
``entanglement-breaking'' channels~\cite{HHR03}, because it turns out that the form
 \eq{MP} is equivalent to the condition that $(T \ot I)$ maps all
 states to separable states.
The optimal measure-and-prepare channel is denoted $\MP_{n\ra k}$ and
is
\be \MP_{n\ra k}(\rho) = \tr_n \bbE_\varphi \dsym n
\varphi^{\ot n+k} (\rho \ot I^{\ot k}) = 
 \tr_n\frac{\dsym n}{\dsym {n+k}} \psym{d}{n+k} (\rho \ot I^{\ot k}).
 \label{eq:opt-MP}\ee

\subsection{Optimal cloning}
The no-cloning theorem says that $\ket{\varphi}\ra \ket{\varphi}\ot
\ket{\varphi}$ is impossible. But in fact $\ket{\varphi}^{\ot n} \ra
\ket{\varphi}^{\ot n+k}$ is also impossible for any $n,k>0$.  Still, we
can try to approximate this map.  It turns out that the optimal
cloning map (due to \cite{Werner98}) is 
\be\clone_{n\ra n+k}(\rho) = \psym{d}{n+k} (\rho \ot I^{\ot k})
\psym{d}{n+k} \frac{\dsym n}{\dsym {n+k}}.
\label{eq:clone}\ee
Note that normally $\rho = \varphi^{\ot n}$.

\eq{clone} is rather remarkable.  At first, it's not even obvious that
$\clone_n^{n+k}$ is trace-preserving, but this can be deduced with the
help of \corref{commutant}.  Optimality takes more work, but is similar
in spirit to the optimality of the above estimation procedure.

The main theorem of \cite{Chiri10} gives a relation between $\MP$ and
$\clone$.  Specifically
\begin{thm}[Chiribella's theorem\cite{Chiri10}]\label{thm:Chiri}
\be \MP_{n\ra k}(\rho) = \sum_{s=0}^k \frac{\binom{n}{s}
  \binom{d+k-1}{k-s}} {\binom{d+n+k-1}{k}} 
\clone_{s\ra  k}(\tr_{n-s}\rho)
\label{eq:MP-clone}\ee
\end{thm}

We give a slightly simpler proof than the one in \cite{Chiri10}.

\begin{proof}
First, we observe (following \cite{Chiri10}) that the space of density
matrices on $\ssym d n$ is spanned by vectors of the form $\proj
\varphi^{\ot n}$.  Thus, to compute the action of $\MP_{n\ra k}$, it
suffices to calculate
$$f(\alpha,\beta) := \tr \beta^{\ot k} \MP_{n\ra k}(\alpha^{\ot n})$$
for all unit vectors $\ket\alpha,\ket\beta\in \bbC^d$. Further, the
unitary covariance of quantities involved means that $f(\alpha,\beta)$
depends only on the scalar $x:=\tr \alpha\beta$.  Let $f(x) := f(\alpha,\beta)$.

Using the definition in \eq{opt-MP}, we see that
\begin{subequations}\label{eq:anal-LHS}
\ba f(x) &= \frac{\dsym n}{\dsym{n+k}}
\tr (I^{\ot n} \ot \beta^{\ot k}) \psym{d}{n+k}
(\alpha^{\ot n} \ot I^{\ot k}) \\
&= \frac{\dsym n}{\dsym{n+k}}
\tr \psym{d}{n+k}(\alpha^{\ot n} \ot \beta^{\ot k}) \\
&= \frac{\dsym n}{\dsym{n+k}}\sum_{s=0}^k
\frac{\binom{k}{s}\binom{n}{s}}{\binom{n+k}{k}}
x^s 
\ea
\end{subequations}
The term on the last line is the probability that a random $\pi \in
\cS_{n+k}$ satisfies $|\pi([n]) \cap [n]| = s$.  This is a hypergeometric
distribution, equivalent to the probability that when $n$ balls are
drawn without replacement from a bucket of $n$ white balls and $k$
black balls, that the resulting sample contains $n-s$ white balls and
$s$ black balls.

On the other hand, to analyze the RHS of \eq{MP-clone}, we calculate
\begin{subequations}\label{eq:anal-RHS}
\ba \tr \beta^{\ot k}\clone_{s\ra  k}(\tr_{n-s}\alpha^{\ot n})
& =\tr \beta^{\ot k}\clone_{s\ra  k}(\alpha^{\ot s})  \\
& =\tr \beta^{\ot k}\psym{d}{k}(\alpha^{\ot s} \ot I^{\ot
  k-s}) \psym{d}{k}\frac{\dsym s}{\dsym k}\\
& =\frac{\dsym s}{\dsym k} x^s
\ea
\end{subequations}

Finally we calculate
\begin{subequations}\label{eq:rearrange}
\ba \frac{\dsym n \dsym k}{\dsym{n+k} \dsym s}
\frac{\binom{k}{s}\binom{n}{s}}{\binom{n+k}{k}}
&= \frac{\binom{d+k-1}{k}\binom{k}{s}}{\binom{d+s-1}{s}}
\cdot \binom{n}{s}\cdot 
\frac{\binom{d+n-1}{n}}{\binom{d+n+k-1}{n+k}\binom{n+k}{k}}
\\ & = \frac{ \binom{d+k-1}{k-s}\binom{n}{s}} 
{\binom{d+n+k-1}{k}}\ea

\end{subequations}
Combining \eq{anal-LHS}, \eq{anal-RHS} and \eq{rearrange}, we obtain \eq{MP-clone}.
\end{proof}

Inspired by Chiribella's theorem, we define the polynomials
\be M_k^{(d,n)}(x) = \sum_{s=0}^k \frac{\binom{n}{s}
  \binom{d+k-1}{k-s}} {\binom{d+n+k-1}{k}} x^s
= \sum_{s=0}^k M^{(d,n)}_{k,s} x^s
\label{eq:M-poly}\ee
The coefficients $M_{k,s}^{(d,n)}$ correspond to a hypergeometric
distribution whose moment-generating
function is given by $M_k^{(d,n)}(e^{t})$.

We observe that these polynomials can be described in terms of Jacobi polynomials as 
\be M_k^{(d,n)}(x) = \frac{(x-1)^k}{\binom{d+n+k-1}{k}} P_k^{(n-k, d-1)}\L(\frac{x+1}{x-1}\R),\ee
and so are orthogonal with respect to the weight
$(1-y)^\alpha(1+y)^\beta\mathrm{d}y$ over $y\in[-1,1]$, where
$y=(x+1)/(x-1)$, $\alpha=n-k$ and $\beta=d-1$.  Unfortunately, if
$x\in [0,1]$, then $y\leq -1$, so this standard interpretation of the
Jacobi polynomials appears not to apply.  Similarly, the
interpretation of Jacobi polynomials as matrix elements of irreps of
$\cU_2$ only applies directly when the argument is in the range
$[-1,1]$.  Jacobi polynomials have previously appeared in analysis of
de Finetti errors in \cite{NOP09}, and it is possible that an
alternate derivation of \eq{M-poly} might proceed via the
representation theory of $\cU_d$ rather than $\cS_n$.  Similar
polynomials have also been analyzed in terms of functions of two
variables~\cite{GrunbaumVZ04}. 

\subsection{de Finetti theorem}
What's the point of all these expansions?  Who cares if we can shuffle
a bunch of permutations around and relate one thing that we didn't
care that much about (the $n\ra k$ measure-and-prepare channel) to the far more
obscure task of choosing $s$ from a hypergeometric distribution,
tracing out all but $s$ subsystems, and cloning back up to $k$?

One application mentioned in \cite{Chiri10} is to give an alternate
proof of the de Finetti theorem.  Observe that
\be M^{(d,n)}_{k,k} = \frac{\binom{n}{k}}{\binom{d+n+k-1}{k}}
= \frac{n!d+n-1!}{n-k!d+n+k-1!} \geq \L(1-\frac{d+k}{n+d}\R)^k
 \geq 1- \frac{k(d+k)}{n+d}
\label{eq:Mdnkk-bound}
\ee
Thus, \eq{MP-clone} implies that
\be \MP_{n\ra k} = (1-\eps)\tr_{n-k} + \eps \cN,\ee
where $\eps \leq k(d+k)/(n+d)$ (assuming this quantity is $\leq 1$)
and $\cN$ is a trace-preserving quantum operation.  Thus
\be \| \MP_{n\ra k} - \tr_{n-k} \|_\diamond \leq 2\eps.\label{eq:deF}\ee
This establishes the de Finetti theorem in an elegant form: given a
symmetric state on $n$ qudits, tracing out $n-k$ qudits yields a state
that is within $2\eps$ of a mixture of tensor powers.  The advantages
of this formula are that it is concise, it naturally handles the case
of symmetric states that are entangled with reference systems and it
gives an explicit description of how to produce the approximation.

In fact, this approach can also yield the so-called
exponential de Finetti theorem of \cite{Ren05,KM09}.  This is the only
original result in this section.

To introduce the exponential de Finetti theorem, we need the idea of
an ``almost-product state'' introduced in \cite{Ren05} (see also \cite{Renner07}).  Define the
$(k,r,d)$-almost product states to be
$$\bigcup_{\ket\varphi\in\bbC^d} \Span \{\psym d k \ket{\varphi}^{\ot
  k-r}\ot\ket\psi : \ket\psi\in(\bbC^d)^{\ot r}\}.$$
This set is {\em not} a linear subspace, but see \cite{KM09} for a discussion of almost-product states from a representation-theoretic perspective.
Note that the set of almost-product states has no real classical analogue, and indeed the
exponential de Finetti theorem (stated below) fails in the classical
case.\footnote{I am grateful to Matthias Christandl and Ben Toner for
  sharing with me their unpublished manuscript which proves this
  point.  The idea of their proof is to compare variances.  If we choose a random sample of $(1-o(1))n$ positions from $0^{n/2}1^{n/2}$, then the resulting distribution of Hamming weights will have $o(n)$ variance.  However, any almost-product distribution that is approximately balanced must have $\Omega(n)$ variance.  With some technical effort, they then translate this into a lower bound on the trace distance between the resulting distributions.
}  Observe that $\clone_{k-s\ra k}$ maps product states to $(k,s,d)$-almost-product
states, and thus $\clone_{k-s\ra k}\circ\MP_{n\ra k-s}$ maps symmetric
states to $(k,s,d)$-almost-product states.
\begin{thm}[Exponential de Finetti theorem\cite{Renner07,KM09}]
For any $0\leq r \leq k$, there exist $x_0,\ldots,x_r\in\bbR$ such
that
$|x_s| \leq (2\delta)^s/(1-\delta)$,  $\eps := \delta^r /
(1-3\delta)$, $\delta := k(d+k)/n$ and 
\be \L\|\tr_{n-k} - \sum_{s=0}^r x_s \clone_{k-s\ra k}\circ\MP_{n\ra
  k-s}\R\|_\diamond \leq \eps,\label{eq:exp-deF}, \ee
where the maps in \eq{exp-deF} are restricted to act on $\vee^n\bbC^d$.
\end{thm}

Perhaps a more natural formulation comes from taking $r=k$, so that
\be \tr_{n-k} = \sum_{s=0}^k x_s \clone_{k-s\ra k}\circ\MP_{n\ra
  k-s}, \label{eq:exp-deF-alt}\ee
with again the bound $|x_s| \leq (2\delta)^s/(1-\delta)$ for each
$s$.  

By contrast, the error in \cite{Renner07} is $\leq 3
(n-k)^d\exp(-\frac{(r+1)(n-k)}{n})$, and \cite{KM09} has a similar
bound.  Our result is thus weaker when $n-k$ is small (say
$\sim n^{2/3}$), but stronger when $r$ is small and $d$ is large.
The likely culprit for this disadvantage is the fact that we upper-bound an
alternating sum by taking the absolute value of each term.

\begin{proof}
The idea is to write $\tr_{n-k}$ as a linear combination of
$\clone_{k-s\ra k}\circ\MP_{n\ra k-s}$ by inverting the formula
\eq{MP-clone}.  For brevity, fix $n,k,d$, let $A_s$ denote
$\clone_{k-s\ra k}\circ\tr_{n-(k-s)}$ and let $B_s$ 
denote $\clone_{k-s\ra k}\circ\MP_{n\ra k-s}$.  In this notation, we have
$B_0 = \sum_{s=0}^k M_{k,k-s}^{(d,n)} A_s$.
Observe also that $\clone_{b\ra  c}\circ\clone_{a\ra b}=\clone_{a\ra c}$.

We now rearrange \eq{MP-clone} to obtain
\be A_0 = \frac{B_0}{M_{k,k}^{(d,n)}} - \sum_{s=1}^{k} 
\frac{M_{k,k-s}^{(d,n)}}{M_{k,k}^{(d,n)}}A_s
\label{eq:invert-MP-clone}\ee
From \eq{Mdnkk-bound} we have that 
\be \frac{1}{M_{k,k}^{(d,n)}} \leq
\L(1-\frac{k(d+k)}{n+d}\R)^{-1} \leq (1-\delta)^{-1}
\label{eq:bound-invMkk}\ee
Similarly,  
\be\frac{M_{k,k-s}^{(d,n)}}{M_{k,k}^{(d,n)}}
 = \frac{\binom{n}{k-s} \binom{d+k-1}{s}}{\binom{n}{k}}
= \binom{k}{s} \frac{(d+k-s)\cdots(d+k-1)}{(n-k)\cdots(n-k+s-1)}
\leq \L(\frac{k(d+k)}{n}\R)^s
 = \delta^s
\label{eq:bound-Mdn}\ee

We now claim that for each $r$, there exists $x_0,\ldots,x_r,y_{r+1}^{(r)},\ldots,y_k^{(r)} \in \bbR$
such that
\begin{subequations}\label{eq:exp-deF-induction}
\be A_0 = \sum_{s=0}^{r-1} x_s B_s + \sum_{s=r}^k y_s^{(r)} A_s
\label{eq:A-expansion}\ee
and the coefficients satisfy 
\be 
|y_s^{(r)}| \leq 2^r \delta^s
\qquad \text{and}\qquad
|x_s| \leq  \frac{|y_s^{(s)}|}{1-\delta} \leq \frac{(2\delta)^s}{1-\delta}
\label{eq:xy-bound}
\ee\end{subequations}
We prove \eq{exp-deF-induction} by induction.  The $r=0$ case is trivial; we simply have $y_0^{(0)}=1$.
Next, for $r\geq 0$, we assume that \eq{exp-deF-induction} holds for $r$ and attempt to prove it for $r+1$. First we replace the $y_r^{(r)} A_r$ term in \eq{A-expansion} with the linear combination of $B_r$ and $A_{r+1},\ldots,A_k$ given by \eq{invert-MP-clone}, to obtain 
\be x_r =  \frac{y_r}{M_{r,r}^{(d,n)}}
\qquad \text{and}\qquad
y^{(r+1)}_s = y^{(r)}_s
- \frac{M_{k-r,k-s}^{(d,n)}}{M_{k-r,k-r}^{(d,n)}} y_{r}^{(r)}\ee
Using \eq{bound-invMkk} we obtain the claimed bound on $|x_r|$ in \eq{xy-bound}.  To obtain the claimed bound on $|y_s^{(r+1)}|$, we use induction and \eq{bound-Mdn} to argue that $|y_s^{(r+1)}| \leq 2^r\delta^s + \delta^{s-r}\cdot 2^r\delta^r = 2^{r+1}\delta^s$.

\end{proof}

\subsubsection{Applications of the de Finetti
  theorem}\label{sec:deF}
The de Finetti theorem has an amazing array of applications, but these
are not entirely obvious upon first inspection.  We wil avoid delving into
them deeply here, but single out only two. 

\benum
\item {\em Extensive quantities.}  Often we are interested in {\em
    extensive} properties of a state, such as energy or entropy, that
  scale linearly with the number of copies of a state.  In other
  words, they satisfy $f(\rho^{\ot n}) = nf(\rho)$.  In this case, the
  de Finetti approximation provides a way to understand extensive
  properties of symmetric states by reducing to the case of density
  matrices on single systems.  See \cite{Renner07} for more discussion
  of this point, \cite{Ren05} for an application to quantum key
  distribution or \cite{fannes06} for an application to mean-field Hamiltonians.
\item {\em Approximating separable states.} The set of separable
  density matrices (i.e. the convex hull of
  $\proj{\alpha}\ot\proj\beta$) is notoriously hard to
  approximate~\cite{Gurvits03, BeigiS10}, so we are forced to use
  heuristics and relaxations.  One of the leading relaxations comes
  from the de Finetti theorem.  We say that $\rho^{AB}$ is
  $k$-extendible if there exists a state $\sigma^{AB_1\cdots B_k}$
  such that $\supp\sigma^{B_1\cdots B_k}\subseteq \vee^k B$ and
  $\rho^{AB} = \sigma^{AB_1}$.  All separable states are clearly
  $k$-extendible for all $k$, and it can also be shown that all
  non-separable (i.e. entangled) states {\em fail} to be
  $k$-extendible for some, perhaps large, $k$.  Thus, the set of
  $k$-extendible states comprises a hierarchy of relaxations of the
  set of separable states~\cite{DohertyPS04}.  This can be understood
  in terms of the fact that by the de Finetti theorem, tracing out
  $B_2\ldots B_k$ is similar to applying the entanglement-breaking
  channel $\MP_{k\ra 1}$ to $B_1\ldots B_k$.  (Alternate intuition
  comes from the idea of ``monogamy of entanglement,'' which states
  that $A$ cannot simultaneously be highly entangled with all of the
  $B_i$; e.g. see \cite{Yang06}.)  An intriguing open problem is to
  understand how $k$-extendability combines with the PPT condition;
  see \cite{NOP09} for some work along these lines.
\eenum

Research on the de Finetti theorem continues, and the interested
reader is referred to \cite{KM09} for a
far-reaching representation-theoretic generalization or
\cite{brandao11} for a powerful variant that uses a different norm
than the trace distance (a line of work continued in \cite{BH-local}).

\section{Netless concentration of measure}\label{sec:conc}

In this section, we show how the symmetric subspace can be a way to
prove large-deviation bounds in a manner analogous to controlling
higher moments of real random variables.  The techniques are (to my
knowledge) new, but the results obtained are not substantially
different from previous results.  One appealing feature of these
results, though, is the unified derivation of previously unrelated
statements about the minimum entanglement of random subspaces.

\subsection{Introduction}
A useful trick in high-dimensional geometry is to combine a concentration-of-measure bound with a union bound over an epsilon net.  The idea is that we have a metric space $S$, a random function $f:S\ra \bbR$ (usually with the Lipschitz property $|f(x)-f(y)|\leq d(x,y)$) and an $\eps$-net $N\subset S$. 
The concentration-of-measure bound states that for any $x$, $f(x)$ has extremely low probability, say $\leq \eta$, of deviating from its mean value $\mu$ by more than $\delta$.  This implies that with probability $\geq 1 - |N|\eta$ (and often we only need that this probability is $>0$),  we have $|f(x)-\mu|\leq \delta$ for all $x\in N$, and by the Lipschitz property, $|f(x)-\mu|\leq \eps+\delta$ everywhere.  In some cases, we can do better.  For example, if $f$ is a seminorm (and $S$ satisfies some more conditions, maybe having diameter 1) then we can obtain the often stronger bound $|f(x)-\mu| \leq \delta / (1-\eps)$.

There is an alternate way to view the first sort of bound.  For an event $E$, define $[E]$ to be the random variable that is 1 if $E$ is true and 0 if $E$ is false.  Then define 
$$g(x) := [f(x) \geq \mu + \delta].$$
By our concentration-of-measure assumption, for any $x$, $\bbE_f[g(x)] \leq \eta$.  Now fix a normalized measure on $S$.  Then $\bbE_x \bbE_f[g(x)] \leq \delta$ as well.  For any $x$, let $B(x,\eps)$ denote the ball of radius $\eps$ around $x$.  Assume further that our measure on $S$ has the property that $|B(x,\eps)|=V(\eps)$, i.e. is independent of $x$.  Now let
 $$p=\Pr\{ \exists \hat x\in S \,:\, f(\hat x)\geq \mu+\eps+\delta\}. $$  
In case such an $\hat x$ exists, the Lipschitz condition on $f$ guarantees that $f(x)\geq \mu+\delta$ for all $x\in B(\hat x,\eps)$.
Therefore we have $\bbE_x \bbE_f [g(x)] \geq p V(\eps)$.  Combining our two bounds on $\bbE_x \bbE_f[ g(x)]$, we find that $p \leq \eta / V(\eps)$.  Since $V(\eps)^{-1} \leq |N| \leq V(\eps/2)^{-1}$ for a minimal $\eps$-net $N$, this yields bounds that are at least as strong as the $\eps$-net-based approach, although not dramatically better.

However, this approach can be further improved by different choices of
function $g$.  Indeed, we need only that $\bbE_x \bbE_f [g(x)]$ is
extremely small, and that conditioned on $f(x)$ being large for some
value of $x$, $\bbE_f [g(x)]$ is also large.  This is the idea behind
Chebyshev's inequality and the Bernstein trick, in which $g(x)$ is
taken to be either $(f(x)-\mu)^2$ or $e^{y f(x)}$, respectively.  The
idea of such choices of $g$ is to amplify large deviations of $f$ so
that they make have a greater effect on the expectation.  These
techniques have been useful in quantum information theory for proving
concentration bounds, for example in \cite{BHLSW03}, where a moment
generating function was used, and in \cite{ADHW06}, where bounding
the second moment was sufficient to produce powerful results. 

One reason that $g(x)=e^{y f(x)}$ is an appealing choice is Cramer's
theorem~\cite{DemboZ09}, which, up to technical caveats, is as
follows. When $x$ is of the form $(x_1,\ldots,x_n)$ for
i.i.d. $x_1,\ldots,x_n$ and with $f(x) := \frac{1}{n}\sum_{i=1}^n
F(x_i)$, then (i) $\Pr\{f(x) \geq a\} \sim \exp(-n s(a) - o(n))$ for
some $s(a)$, and (ii) optimizing over $y$ can yield the nearly-optimal
bound of $\exp(-ns(a))$.  (To relate to the earlier discussion, we
have $a = \mu+\delta$.) 
Indeed, $s(a) = \sup_y (ya - \ln\bbE_{x_1} [e^{yF(x_1)}])$.  

However, it turns out that taking $g(x)= x^p$ and optimizing over $p$
always yields a bound that is at least as powerful than when $g(x)$ is
of the form $e^{yf(x)}$~\cite{DubhashiP09}. For this to work, we need
that $f(x)\geq 0$ with probability 1, but no longer need the
i.i.d. assumption. To see why $g(x)=x^p$ is at least as good a choice,
let $\gamma(a) = \min_{p\in \bbN} \bbE[f(x)^p] / a^p$ be the optimal
bound obtainable by optimizing over $p$.  Then $\bbE[f(x)^p] \geq
\gamma(a)\cdot a^p$ for all nonnegative integers $p$, and thus
$\bbE[e^{yx}] \geq \sum_{p\geq 0} \gamma(a)\, y^pa^p/p! = \gamma(a)\,
e^{ya}$ and finally $e^{-s(a)}\geq \gamma(a)$. 

In this section, we will focus on showing that random subspaces are
likely to contain only highly entangled states.  Thus our results will
be similar in many ways to those of \cite{HLW06}, which used the more
conventional methods of $\eps$-nets and Levy's Lemma (which is based
on Gaussian concentration, which in turn can be derived from
moment-generating functions).  The advantages of this approach is that
the proof is somewhat more self-contained and the resulting bounds are
now strong ehough to unify several different previous results.  The
main disadvantage compared with Levy's Lemma is a loss in flexibility,
a limitation whose consequences we will return to below.

\subsection{Statement of results}
Define $S^d$ to be the set of unit vectors in $\bbC^d$.   All
expectations are taken with respect to unitarily invariant measures.  
In this part of the paper, we will always consider the
following scenario.  There are $k$ quantum systems of dimensions $d_1,
d_2, \ldots d_k$, with $D := d_1d_2\cdots d_k$.  For any Hermitian operator
$\Pi$ acting on $\bbC^D$, we will define
$$\nu(\Pi) := \max \L\{ \tr (\varphi_1 \ot \cdots \ot \varphi_k)\Pi :
\ket{\varphi_1}\in S^{d_1},\ldots,\ket{\varphi_k}\in S^{d_k}\R\}.$$
We will generally consider the case when $\Pi$ is a random orthogonal
projector of rank $r$.

\begin{thm}\label{thm:general}
Let $\Pi$ be a random rank-$r$ orthogonal projector acting on
$\bigotimes_{i=1}^k \bbC^{d_i}$.  Then for any $\gamma>0$, 
\be \Pr_\Pi \L[\nu(\Pi)\geq \gamma\R] \leq \inf_n 
\frac{\binom{r + n - 1}{n}
\prod_{i=1}^k \binom{d_i + n - 1}{n}}
{\gamma^n \binom{D + n-1}{n}}
\label{eq:orig-product-bound}\ee
\end{thm}

Before presenting the proof, we examine three corollaries of
\thmref{general} corresponding to different special cases.

\subsubsection{Large subspaces}
One limit is the case of subspaces with small codimension, where the
minimal entanglement is small or zero.
\begin{cor}\label{cor:no-prod-states}
Let $D=\prod_{i=1}^k d_i$ and let $V$ be a uniformly random
projector in $\bigotimes_{i=1}^k \bbC^{d_i}$ of rank $r$ such that $D
> r + \sum_{i=1}^k (d_i - 1)$. Then the probability that $V$ contains
a product state is zero.  Equivalently, if we take $\Pi$ to be the
orthogonal projector onto $V$, then  $\nu(\Pi)<1$ with probability 1.
\end{cor}
This can be proven by standard algebraic-geometric
arguments~\cite{Eisenbud88,IlicL99}; a more explicit argument for
this fact was given recently by Walgate and Scott~\cite{WalgateS07}.
These works also proved the optimality of \corref{no-prod-states},
meaning that if $D \leq r + \sum_{i=1}^k 
(d_i - 1)$ then any subspace of dimension $r$ must contain at least
one product state.  

\begin{proof}  Set $\gamma=1$.  Then the RHS of
\eq{orig-product-bound} is 
\be \Pr_\Pi \L[\nu(\Pi)\geq \gamma\R] \leq \inf_n 
\frac{\binom{r + n - 1}{n}
\prod_{i=1}^k \binom{d_i + n - 1}{n}}
{ \binom{D + n-1}{n}}
\label{eq:delta-bound}\ee
Note that when $d$ is fixed and  $n$ is large $\binom{d+n-1}{n}=
O(n^{d-1})$.  Thus as $n\ra\infty$ \eq{delta-bound} is
$$O(n^{\sum_{i=1}^k(d_i-1)+r-1-D+1})$$
which tends to zero if $D > r + \sum_{i=1}^k (d_i - 1)$.
\end{proof}

One difference between our proof and those based on algebraic geometry
is that ours degrades smoothly when we take $\gamma$ to be slightly
smaller than one.  Indeed, we can prove a nonzero, but weak, lower
bound on the minimum entanglement of vector spaces meeting
the conditions of \corref{no-prod-states}.  For simplicity, we
consider the case of $k=2$ and $d_1=d_2=d$, although the general case
poses no additional difficulties.

\begin{prop}
Let $\Pi$ be the projectors onto a random subspace of $\bbC^d \ot \bbC^d$ of dimension
$r=d^2 - 2(d-1) - x$ for some positive integer $x$.  Then
\be \Pr_\Pi [\nu(\Pi) \geq 1-d^{-2-2d/x}] \leq d^{-d}
\ee
\end{prop}
\begin{proof}
Set $n=d^{2+2d/x}$ so that $\gamma = 1-1/n$.  Observe that 
\be \frac{n^{a-1}}{a-1!}\leq 
\binom{n+a-1}{n} \leq \frac{n^{a-1}}{a-1!}e^{a^2/2n}.\ee
Applying this to \eq{orig-product-bound} with $\gamma=$, we obtain
\bas \Pr_\Pi[\nu(\Pi)\geq \gamma] &\leq
\frac{d^2-1!}{d-1!^2 (d^2-2(d-1)-x)!} \frac{1}{\gamma^nn^x} e^{d^4/n}
\\ &\leq \frac{d^{2(d-1)+2x}}{\gamma^nn^x}
\eas
Substituting $\gamma = 1-1/n$ yields the desired bound.
\end{proof}

\subsubsection{Entanglement of random pure bipartite states}
Second, consider the case when $r=1$, and so $\Pi=\proj\psi$.  When
$k=2$, $\nu(\psi)$ is simply the largest Schmidt value of a random
state, and should be roughly $(1/\sqrt{d_1}+1/\sqrt{d_2})^2$, according
to the Mar\v{c}enko-Pastur law~\cite{GT04}.  Unfortunately, \thmref{general} is
not quite strong enough to prove this, and can only obtain a bound of
$1/d_1+\frac{O(1+\log(d_2/d_1))}{\sqrt{d_1d_2}}$ when $d_1 \gg d_2$.
To illustrate the technique, we consider the case of $d_1=d_2=d$, when
the true value of $\nu(\Pi)$ is $\approx 4/d$, and we achieve a result
that is weaker by a constant factor.
\begin{cor}\label{cor:entangled}
 Let $\ket{\psi}$ be drawn uniformly at random from
$\bbC^{d} \ot \bbC^{d}$,
Let $\gamma_0 = \frac{16}{ed}$.  Then
\be
\Pr_\psi \left\{\|\psi^A\|_\infty \geq \gamma_0 e^\eps\right\}
 \leq e^{-d\eps}
\label{eq:pure-state-bound}\ee
\end{cor}

\begin{proof}
Let $n=d$. Then $\binom{d+n-1}{n} \leq 4^d$ and $\binom{d^2+n-1}{n}
\geq d^{2d}/d! \geq (ed)^d$.  We now substitute into
\eq{orig-product-bound} and obtain \eq{pure-state-bound}.
\end{proof}

\subsubsection{Multi-qubit states}
We also recover another corollary in the case of many qubits that
slightly sharpens the main technical result of \cite{GFE09}.  
\begin{cor}
Let
$d_1=\cdots=d_k=2$ and $r=1$.  Choose some $\eps>0$.  Then 
a random $k$-qubit pure state has probability 
\be \L(\frac{1}{\eps^{(1+1/\eps)}k}\R)^k \sim k^{-k} \ee
of having overlap $\geq \gamma := k^{1+2\eps}2^{-k}/e$ with any
$k$-qubit product state.
\end{cor}

\begin{proof}
Plugging $d_1=\cdots=d_k=2$ and $r=1$ into \eq{orig-product-bound}
yields an upper bound of 
\be \frac{(n+1)^k n!}{\gamma^n 2^k(2^k+1)\cdots (2^k+n-1)}
\label{eq:geom-ent}\ee
We then choose $n=k/\eps$ and can bound \eq{geom-ent} with
$$\leq \frac{(n+1)^k n!}{(k^{1+2\eps}/e)^n}
\leq \frac{ (k/\eps)^k (k/e\eps)^{k/\eps}}
{
 k^{k(2+1/\eps)}e^{k/\eps}}
= \L(\frac{1}{\eps^{(1+1/\eps)}k}\R)^k.
$$
\end{proof}

By contrast, Gross, Flammia and Eisert~\cite{GFE09} prove that the
probability of $\nu(\Pi) 
\geq 8k^22^{-k}$ is $\leq e^{-k^2}$.

\subsection{Proof of the main result}
To prove \thmref{general}, we will relate the maximum overlap of $\Pi$
with a product state to the $n^{\text{th}}$ moment of its overlap,
which we define as
\be \mu_k^n(\Pi) := \bbE_{\varphi_1,\ldots,\varphi_k} 
(\tr (\Pi \bigotimes \varphi_i))^n
\label{eq:mu-def}\ee
Naturally $\nu(\Pi) = \lim_{n\ra\infty} \mu_k^n(\Pi)^{1/n}$.  But we
will see that more precise quantitative estimates are possible.

It will be convenient to let \eq{mu-def} be defined for any positive
semidefinite operator $\Pi$.
\begin{lem}\label{lem:mu}\ 

\bit
 \item $\mu_k^n$ is {\em homogenous}.  That is, for any $x>0$ and
$\Pi\geq 0$, $\mu_k^n(x\Pi) = x^n \mu_k^n(\Pi)$.
\item $\mu_k^n$ is {\em non-decreasing}.  That is, if $0\leq A\leq B$, then
  $0\leq \mu_k^n(A) \leq \mu_k^n(B)$.
\eit
\end{lem}
The proof is omitted.  (In fact, the $\mu_k^n$ are norms.  See
\cite{Montanaro12} for more discussions of their properties and
relations between $\mu_k^n$ for different values of $n$.)

The strategy of our proof is to calculate $\bbE_\Pi[\mu_k^n(\Pi)]$ in
two different ways.  On the one hand, it can be evaluated exactly, as
we will discuss below.  On the other hand, for any fixed $\Pi$, this
expression can be lower-bounded in terms of $\nu(\Pi)$ as follows:

\begin{lem}\label{lem:prod-moment}
For any $\Pi$ and any $n>0$,
\be \mu_k^n(\Pi)
\geq \frac{\nu(\Pi)^n}{\prod_{i=1}^k \binom{d_i+n-1}{n}}.
\label{eq:prod-moment}\ee
\end{lem}

\begin{proof}[Proof of \lemref{prod-moment}]
Let $\ket{\hat\varphi}=\ket{\hat\varphi_1}\ot\cdots \ot
\ket{\hat\varphi_k}$ be a product state maximizing $\tr 
\Pi\hat\varphi$; i.e. such that $\tr\Pi\hat\varphi=\nu(\Pi)$.  Define $p :=
\tr\Pi\hat\varphi$ and $\ket{\psi} := p^{-1/2} \Pi\hat\varphi$.
Note that $\Pi\geq \psi$ and that $\nu(\psi)=\nu(\Pi)$.  Thus, it
suffices to prove the lemma in the case when $\Pi = \psi$.

Let $\ket{\psi_k}=\ket\psi$.  We now iteratively define
$p_{k-1},\ldots,p_1$ and 
$\ket{\psi_{k-1}},\ldots,\ket{\psi_1}$ as follows.  For
$j=k-1,\ldots,1$, choose $p_j>0$ and $\ket{\psi_j}\in S^{d_1\cdots
  d_j}$ to satisfy
\be
I_{d_1} \ot \cdots \ot I_{d_{j-1}} \ot \bra{\hat\varphi_j} \cdot
\ket{\psi_j} = \sqrt{p_j} \ket{\psi_j}.\ee
Observe that $p=p_1\cdots p_{k-1}$.

We will show that 
\be \mu_k^n(\psi) \geq \frac{p_{k-1}^n}{\binom{d_k+n-1}{n}}
\mu_{k-1}^n(\psi_{k-1})
\label{eq:prod-moment-recur}.\ee
This can then be applied inductively to establish the lemma.

Now we calculate 
\ba  
\mu_k^n(\Pi) & = 
\bbE_{\varphi_1, \ldots, \varphi_k} 
\tr \psi^{\ot n} \bigotimes_{i=1}^k \varphi_i^{\ot n}
\\& = 
\bbE_{\varphi_1, \ldots, \varphi_{k-1}} \sum_{\pi\in\cS_n}
\tr \psi^{\ot n} \L( 
\bigotimes_{i=1}^{k-1} \varphi_i^{\ot n} \ot
\frac{ P_d(\pi)}{d_k^{\bar n}} \R)
& \text{by \propref{twirl}} \\
& = 
\bbE_{\varphi_1, \ldots, \varphi_{k-1}} \sum_{\pi\in\cS_n}
\tr \psi^{\ot n} 
P_d(\pi)^{\ot k} P_d(\pi^{-1})^{\ot k}
\L( \bigotimes_{i=1}^{k-1} \varphi_i^{\ot n} \ot
\frac{ P_d(\pi)}{d_k^{\bar n}} \R)\hspace{-2cm}
\\& = 
\bbE_{\varphi_1, \ldots, \varphi_{k-1}} \sum_{\pi\in\cS_n}
\tr \psi^{\ot n} 
\L( \bigotimes_{i=1}^{k-1} \varphi_i^{\ot n} \ot
\frac{I_d^{\ot n}}{d_k^{\bar n}} \R)
& \text{since $\bra{\psi}^{\ot n}$ and
  $\ket{\varphi_i}^{\ot n}$ are symmetric}
\\ &=
\frac{1}{\binom{d_k+n-1}{n}}\bbE_{\varphi_1, \ldots, \varphi_{k-1}}
\tr (\tr_k\psi)^{\ot n}  \bigotimes_{i=1}^{k-1} \varphi_i^{\ot n} 
\\ & = \frac{1}{\binom{d_k+n-1}{n}} \mu_{k-1}^n(\tr_k\psi)
\\ & \geq \frac{1}{\binom{d_k+n-1}{n}} \mu_{k-1}^n(p_{k-1}\psi_{k-1})
& \text{since }\tr_k\psi \geq p_{k-1}\psi_{k-1}
\\ & = \frac{p_{k-1}^n}{\binom{d_k+n-1}{n}} \mu_{k-1}^n(\psi_{k-1})
& \text{by homogeneity (\lemref{mu})}
\ea
This concludes the proof of the Lemma.
\end{proof}

Remark: \lemref{prod-moment} has the following alternate
interpretation (which we will make use of).
\begin{multline}
\max\{|\braket{\psi}{\hat\varphi_1,\ldots,\hat\varphi_k}|^{2n}  :
\ket{\hat\varphi_1}\in S^{d_1},\ldots,\ket{\hat\varphi_k}\in S^{d_k}\}
\cdot
\bbE_{\ket{\varphi_1}\in S^{d_1},\ldots,\ket{\varphi_k}\in S^{d_k}}
[|\braket{\hat\varphi_1,\ldots,\hat\varphi_k}
{\hat\varphi_1,\ldots,\hat\varphi_k}|^{2n}]
\\ \leq
\bbE_{\ket{\varphi_1}\in S^{d_1},\ldots,\ket{\varphi_k}\in S^{d_k}}
[|\braket{\psi}{\hat\varphi_1,\ldots,\hat\varphi_k}|^{2n}].
\end{multline}

\begin{multline}
\max\{|\braket{\psi}{\hat\varphi_1,\ldots,\hat\varphi_k}|^{2n}  :
\ket{\hat\varphi_1}\in S^{d_1},\ldots,\ket{\hat\varphi_k}\in S^{d_k}\}
\\  \leq
\frac{\bbE_{\ket{\varphi_1}\in S^{d_1},\ldots,\ket{\varphi_k}\in S^{d_k}}
[|\braket{\psi}{\hat\varphi_1,\ldots,\hat\varphi_k}|^{2n}] \hfill}
{\bbE_{\ket{\varphi_1}\in S^{d_1},\ldots,\ket{\varphi_k}\in S^{d_k}}
[|\braket{\varphi_1,\ldots,\varphi_k}
{\hat\varphi_1,\ldots,\hat\varphi_k}|^{2n}]}
\end{multline}

\begin{proof}[Proof of \thmref{general}]
Let $S_\gamma := \{\Pi : \nu(\Pi) \geq \gamma\}$ and let $p := \Pr_\Pi\{\Pi \in S_\gamma\}$.  Our goal is to upper
bound $p$.  We will do this by computing
\be \bbE_\Pi \mu_k^n(\Pi) \label{eq:nth-moment}\ee
in two different ways.

Since $\tr \Pi\varphi$ is always $\geq 0$, we can lower bound the
expectation over all $\Pi$ by considering the contribution only from
$\Pi\in S_\gamma$.  By \lemref{prod-moment} this gives us the lower
bound
\be \bbE_\Pi \mu_k^n(\Pi)
\geq p\frac{\nu(\Pi)^n}{\prod_{i=1}^k \binom{d_i+n-1}{n}}.
\label{eq:moment-bound}\ee

On the other hand, we can also calculate \eq{nth-moment} exactly.
Indeed, $\bbE_\Pi \mu_k^n(\Pi) = \bbE_{\Pi,\varphi} (\tr
\Pi\varphi)^n$ and it turns out that this expectation is independent
of $\varphi$. To see this, let $\Pi=U^\dag\Pi_0 U$ for $\Pi_0$ a fixed
rank-$r$ projector and $U$ drawn uniformly randomly from
$U(D)$. 
\begin{align} \bbE_{\Pi} (\tr \Pi \varphi)^{\ot n}
& = \bbE_{U} (\tr U\varphi U^\dag \Pi_0)^n\non\\
& = \bbE_{U}
\tr (U\varphi U^\dag)^{\ot n} \Pi_0^{\ot n}\non\\
& =  \tr
\frac{P_{\text{sym}}^{D,n}}{\tr P_{\text{sym}}^{D,n}}
 \Pi_0^{\ot n}\non\\
& = \frac{\tr P_{\text{sym}}^{r,n}}{\tr P_{\text{sym}}^{D,n}}
 = \frac{\binom{r+n-1}{r-1}}{\binom{D+n-1}{D-1}}
\label{eq:moment-eq}
\end{align}
Since \eq{moment-eq} holds for all $\varphi$, it also equals the
expectation and in turn equals $\bbE_\Pi \mu_k^n(\Pi)$.  Finally, we
combine \eq{moment-bound} and \eq{moment-eq} to obtain the desired
bound on $p$.
\end{proof}

\subsection{Discussion}
This approach has its strengths, but is also more limited in scope
than techniques based on Levy's Lemma.  For example, replacing the maximum
overlap with product states with some other measure of entanglement
would require more effort.  Even showing the concentration of the {\em
  smallest} Schmidt value of all pure states in a random substate
appears to require some additional ideas, although this is not
completely hopeless.

We remark that these techniques have some significant overlap with the
classic {\em method of moments} from random matrix theory (see
\cite{Tao12, AndersonGZ09} for reviews, or \cite{AHH09} for a quantum example).

\section*{Acknowledgments}
Thanks to Guilio Chiribella, Graeme Mitchison, Michael Walter and
Kevin Zatloukal for helpful comments and discussion. 
I was funded by NSF grants CCF-0916400 and CCF-1111382 and ARO contract
W911NF-12-1-0486.

\bibliographystyle{../latex/hyperabbrv}
\bibliography{../latex/hbib}

\begin{thebibliography}{10}

\bibitem{ADHW06}
A.~Abeyesinghe, I.~Devetak, P.~Hayden, and A.~Winter.
\newblock The mother of all protocols: Restructuring quantum information's
  family tree.
\newblock {\em Proc. Roc. Soc. A}, 465(2108):2537--2563, 2009,
  \href{http://arxiv.org/abs/quant-ph/0606225}{{\ttfamily
  arXiv:quant-ph/0606225}}.

\bibitem{AHH09}
A.~Ambainis, A.~W. Harrow, and M.~Hastings.
\newblock Random tensor theory: extending random matrix theory to random
  product states.
\newblock {\em Commun. Math. Phys.}, 310(1):25--74, 2012,
  \href{http://arxiv.org/abs/0910.0472}{{\ttfamily arXiv:0910.0472}}.

\bibitem{AndersonGZ09}
G.~Anderson, A.~Guionnet, and O.~Zeitouni.
\newblock {\em An Introduction to Random Matrices}.
\newblock Cambridge Studies in Advanced Mathematics. Cambridge University
  Press, 2009.

\bibitem{Audenaert06}
K.~M. Audenaert.
\newblock A digest on representation theory of the symmetric group, 2006.
\newblock {\tt
  http://personal.rhul.ac.uk/usah/080/QITNotes\_files/Irreps\_v06.pdf}.

\bibitem{BBDEJM96}
A.~Barenco, B.~Andr\'{e}, D.~Deutsch, A.~Ekert, R.~Jozsa, and C.~Macchiavello.
\newblock Stabilization of quantum computations by symmetrization.
\newblock {\em SIAM J. Comput.}, 26:1541--1557, 1997,
  \href{http://arxiv.org/abs/quant-ph/9604028}{{\ttfamily
  arXiv:quant-ph/9604028}}.

\bibitem{BeigiS10}
S.~Beigi and P.~W. Shor.
\newblock Approximating the set of separable states using the positive partial
  transpose test.
\newblock {\em J. Math. Phys.}, 51(4):042202, 2010,
  \href{http://arxiv.org/abs/0902.1806}{{\ttfamily arXiv:0902.1806}}.

\bibitem{BHLSW03}
C.~H. Bennett, P.~Hayden, D.~W. Leung, P.~W. Shor, and A.~J. Winter.
\newblock Remote preparation of quantum states.
\newblock {\em IEEE Trans. Inf. Theory}, 51(1):56--74, 2005,
  \href{http://arxiv.org/abs/quant-ph/0307100}{{\ttfamily
  arXiv:quant-ph/0307100}}.

\bibitem{BH-local}
F.~G. Brandao and A.~W. Harrow.
\newblock Quantum de {F}inetti theorems under local measurements with
  applications.
\newblock In {\em Proceedings of the 45th annual ACM symposium on Symposium on
  theory of computing}, STOC '13, pages 861--870, 2013,
  \href{http://arxiv.org/abs/1210.6367}{{\ttfamily arXiv:1210.6367}}.

\bibitem{brandao11}
F.~G. S.~L. Brand{\~{a}}o, M.~Christandl, and J.~Yard.
\newblock Faithful squashed entanglement.
\newblock {\em Commun. Math. Phys.}, 306(3):805--830, 2011,
  \href{http://arxiv.org/abs/1010.1750}{{\ttfamily arXiv:1010.1750}}.

\bibitem{Chiri10}
G.~Chiribella.
\newblock On quantum estimation, quantum cloning and finite quantum de
  {Finetti} theorems.
\newblock In {\em Proceedings of the 5th conference on Theory of quantum
  computation, communication, and cryptography}, TQC'10, pages 9--25, Berlin,
  Heidelberg, 2011. Springer-Verlag,
  \href{http://arxiv.org/abs/1010.1875}{{\ttfamily arXiv:1010.1875}}.

\bibitem{matthias}
M.~Christandl.
\newblock {\em The structure of bipartite quantum states: Insights from group
  theory and cryptography}.
\newblock PhD thesis, University of Cambridge, 2006,
  \href{http://arxiv.org/abs/quant-ph/0604183}{{\ttfamily
  arXiv:quant-ph/0604183}}.

\bibitem{DemboZ09}
A.~Dembo and O.~Zeitouni.
\newblock {\em Large Deviations Techniques and Applications}.
\newblock Stochastic Modelling and Applied Probability. Springer, 2009.

\bibitem{DohertyPS04}
A.~C. Doherty, P.~A. Parrilo, and F.~M. Spedalieri.
\newblock Complete family of separability criteria.
\newblock {\em Phys. Rev. A}, 69:022308, Feb 2004,
  \href{http://arxiv.org/abs/quant-ph/0308032}{{\ttfamily
  arXiv:quant-ph/0308032}}.

\bibitem{DubhashiP09}
D.~Dubhashi and A.~Panconesi.
\newblock {\em Concentration of Measure for the Analysis of Randomized
  Algorithms}.
\newblock Cambridge University Press, 2009.

\bibitem{Eisenbud88}
D.~Eisenbud.
\newblock Linear sections of determinantal varieties.
\newblock {\em Amer. J. Math}, 110(3):541--575, 1988.

\bibitem{fannes06}
M.~Fannes and C.~Vandenplas.
\newblock Finite size mean-field models.
\newblock {\em J. Phys. A}, 39(45):13843, 2006,
  \href{http://arxiv.org/abs/quant-ph/0605216}{{\ttfamily
  arXiv:quant-ph/0605216}}.

\bibitem{folland01}
G.~B. Folland.
\newblock How to integrate a polynomial over a sphere.
\newblock {\em The American Mathematical Monthly}, 108(5):446--448, May 2001.

\bibitem{GW98}
R.~Goodman and N.~Wallach.
\newblock {\em Representations and Invariants of the Classical Groups}.
\newblock Cambridge University Press, 1998.

\bibitem{GT04}
F.~G{\"o}tze and A.~Tikhomirov.
\newblock Rate of convergence in probability to the {Mar\v{c}enko-Pastur} law.
\newblock {\em Bernoulli}, 10(3):503--548, 2004,
  \href{http://arxiv.org/abs/1110.1284}{{\ttfamily arXiv:1110.1284}}.

\bibitem{GFE09}
D.~Gross, S.~T. Flammia, and J.~Eisert.
\newblock Most quantum states are too entangled to be useful as computational
  resources.
\newblock {\em Phys. Rev. Lett.}, 102:190501, May 2009,
  \href{http://arxiv.org/abs/0810.4331}{{\ttfamily arXiv:0810.4331}}.

\bibitem{GrunbaumVZ04}
F.~A. {Gr\"unbaum}, L.~Vinet, and A.~Zhedanov.
\newblock Linear operator pencils on lie algebras and laurent biorthogonal
  polynomials.
\newblock {\em Journal of Physics A: Mathematical and General}, 37(31):7711,
  2004.

\bibitem{Gurvits03}
L.~Gurvits.
\newblock Classical deterministic complexity of {E}dmonds' problem and quantum
  entanglement.
\newblock In {\em Proc. 35\textsuperscript{th} Annual ACM Symp. Theory of
  Computing}, pages 10--19, 2003.
\newblock {arXiv:quant-ph/0303055}.

\bibitem{Har05}
A.~W. Harrow.
\newblock {\em Applications of coherent classical communication and {S}chur
  duality to quantum information theory}.
\newblock PhD thesis, M.I.T., Cambridge, MA, 2005,
  \href{http://arxiv.org/abs/quant-ph/0512255}{{\ttfamily
  arXiv:quant-ph/0512255}}.

\bibitem{HLW06}
P.~Hayden, D.~W. Leung, and A.~Winter.
\newblock Aspects of generic entanglement.
\newblock {\em Commun. Math. Phys.}, 265:95, 2006,
  \href{http://arxiv.org/abs/quant-ph/0407049}{{\ttfamily
  arXiv:quant-ph/0407049}}.

\bibitem{HHR03}
M.~Horodecki, P.~W. Shor, and M.~B. Ruskai.
\newblock General entanglement breaking channels.
\newblock {\em Reviews in Mathematical Physics}, 15:629--641, 2003,
  \href{http://arxiv.org/abs/quant-ph/0302031}{{\ttfamily
  arXiv:quant-ph/0302031}}.

\bibitem{IlicL99}
B.~Ilic and J.~M. Landsberg.
\newblock On symmetric degeneracy loci, spaces of symmetric matrices of
  constant rank and dual varieties.
\newblock {\em Math. Ann.}, 314:159--174, 1999.

\bibitem{JHHH98}
R.~Jozsa, M.~Horodecki, P.~Horodecki, and R.~Horodecki.
\newblock Universal quantum information compression.
\newblock {\em Phys.~Rev.~Lett.}, 81:1714--1717, 1998,
  \href{http://arxiv.org/abs/quant-ph/9805017}{{\ttfamily
  arXiv:quant-ph/9805017}}.

\bibitem{KM09}
R.~Koenig and G.~Mitchison.
\newblock A most compendious and facile quantum de {Finetti} theorem.
\newblock {\em J. Math. Phys.}, 50:012105, 2009,
  \href{http://arxiv.org/abs/quant-ph/0703210}{{\ttfamily
  arXiv:quant-ph/0703210}}.

\bibitem{Montanaro12}
A.~Montanaro.
\newblock Some applications of hypercontractive inequalities in quantum
  information theory.
\newblock {\em Journal of Mathematical Physics}, 53(12):122206, 2012,
  \href{http://arxiv.org/abs/1208.0161}{{\ttfamily arXiv:1208.0161}}.

\bibitem{NOP09}
M.~Navascu\'es, M.~Owari, and M.~B. Plenio.
\newblock Power of symmetric extensions for entanglement detection.
\newblock {\em Phys. Rev. A}, 80:052306, Nov 2009,
  \href{http://arxiv.org/abs/0906.2731}{{\ttfamily arXiv:0906.2731}}.

\bibitem{Ren05}
R.~Renner.
\newblock {\em Security of quantum key distribution}.
\newblock PhD thesis, ETHZ, Zurich, 2005,
  \href{http://arxiv.org/abs/quant-ph/0512258}{{\ttfamily
  arXiv:quant-ph/0512258}}.

\bibitem{Renner07}
R.~Renner.
\newblock Symmetry of large physical systems implies independence of
  subsystems.
\newblock {\em Nature Physics}, 3:645--649, 2007,
  \href{http://arxiv.org/abs/quant-ph/0703069}{{\ttfamily
  arXiv:quant-ph/0703069}}.

\bibitem{Stanley-EC2}
R.~P. Stanley.
\newblock {\em Enumerative Combinatorics, vol. 2}.
\newblock Cambridge University Press, 1999.

\bibitem{Tao12}
T.~Tao.
\newblock {\em Topics in Random Matrix Theory}.
\newblock Graduate Studies in Mathematics. American Mathematical Society, 2012.

\bibitem{WalgateS07}
J.~Walgate and A.~J. Scott.
\newblock Generic local distinguishability and completely entangled subspaces,
  2007,  \href{http://arxiv.org/abs/0709.4238}{{\ttfamily arXiv:0709.4238}}.

\bibitem{Werner98}
R.~F. Werner.
\newblock Optimal cloning of pure states.
\newblock {\em Phys. Rev. A}, 58(3):1827--1832, Sep 1998,
  \href{http://arxiv.org/abs/quant-ph/9804001}{{\ttfamily
  arXiv:quant-ph/9804001}}.

\bibitem{Yang06}
D.~Yang.
\newblock A simple proof of monogamy of entanglement.
\newblock {\em Physics Letters A}, 360(2):249--250, 2006,
  \href{http://arxiv.org/abs/quant-ph/0604168}{{\ttfamily
  arXiv:quant-ph/0604168}}.

\end{thebibliography}
\end{document}